\documentclass[review,onefignum,onetabnum]{siamart171218}
\usepackage{multirow}
\usepackage{amsmath}
\usepackage{color}

\usepackage{amssymb}
\usepackage{mathrsfs}
\def\Q{\mathbf{Q}}
\def\P{\mathbf{P}}
\def\I{\mathbf{I}}

\def\n{\mathbf{n}}
\def\m{\mathbf{m}}
\def\r{\mathbf{r}}

\def\_v{\mathbf{v}}

\usepackage{lipsum}
\usepackage{amsfonts}
\usepackage{graphicx}
\usepackage{epstopdf}
\usepackage{algorithmic}
\newsiamremark{remark}{Remark}
\newsiamremark{hypothesis}{Hypothesis}
\crefname{hypothesis}{Hypothesis}{Hypotheses}
\newsiamthm{claim}{Claim}
\newcommand{\TheTitle}{Nematic liquid crystals in rectangular confinement}
\newcommand{\TheAuthors}{B. Shi, Y. Han, L. Zhang}

\headers{\TheTitle}{\TheAuthors}
\title{{\TheTitle}\thanks{This work was funded by the National Natural Science Foundation of China No. 12050002.}}

\author{Baoming Shi\thanks{School of Mathematical Sciences, Peking University, Beijing 100871, China (\email{ming123@stu.pku.edu.cn}).}
\and  Yucen Han\thanks{Department of Mathematics and Statistics, University of Strathclyde, G1 1XQ, UK. (\email{yucen.han@strath.ac.uk}).}
\and Lei Zhang\thanks{Beijing International Center for Mathematical Research, Center for Quantitative Biology, Peking University, Beijing 100871, China (\email{zhangl@math.pku.edu.cn}).}
}
\ifpdf
\hypersetup{
  pdftitle={\TheTitle},
  pdfauthor={\TheAuthors}
}
\fi
\renewcommand{\TheTitle}{Nematic liquid crystals in a rectangular confinement: solution landscape and bifurcation}

\begin{document}

\maketitle

% REQUIRED
\begin{abstract}
We study the solution landscape and bifurcation diagrams of nematic liquid crystals confined on a rectangle, using a reduced two-dimensional Landau--de Gennes framework in terms of two geometry-dependent variables: half short edge length $\lambda$ and aspect ratio $b$. 
First, we analytically prove that, for any $b$ with a small enough $\lambda$ or for a large enough $b$ with a fixed domain size, there is a unique stable solution that has two line defects on the opposite short edges. Second, we numerically construct solution landscapes by varying $\lambda$ and $b$, and report a novel X state, which emerges from saddle-node bifurcation and serves as the parent state in such a solution landscape. Various new classes are then found among these solution landscapes, including the X class, the S class, and the L class. By tracking the Morse indices of individual solutions, we present bifurcation diagrams for nematic equilibria, thus illustrating the emergence mechanism of critical points and several effects of geometrical anisotropy on confined defect patterns.

\end{abstract}

% REQUIRED

% REQUIRED

\section{Introduction} 
A nematic liquid crystal (NLC) is a viscoelastic anisotropic material that exists between solid and liquid phases of matter and retains an orientational order but has no positional order \cite{de1993physics}. Its moleculars tend to align along certain locally preferred directions, referred to as ``directors'' in the literature. Consequently, NLCs have direction-dependent physical, optical and rheological properties \cite{sonin2018pierre,stewart2019static}, which are applicable to fields such as nanoscience, biophysics, and material design. NLCs usually exhibit topological defects, which are regions of reduced orientational order where the nematic directors cannot be defined. Defects commonly exist as line defects and point defects, and are further classified in topological degrees as $\pm 1$ and $\pm 1/2$ point defects. Although defects are energetically unfavorable, they are unavoidable under geometric confinement. On the application side, defect structures can produce different optical properties and thus play an important role in designing self-assembly structures and colloidal suspensions \cite{foffano2014dynamics,miller2014design}.

In recent decades, several mathematical theories for NLCs, from microscopic models to macroscopic models, have been proposed \cite{de1993physics,doi1988theory,wang2021modeling}. Microscopic Onsager models with different potential kernels, for instance, have been applied to describe the static and dynamic phenomena of liquid crystals \cite{zhang2012onsager,liang2014rigid,liu2005axial,yao2018topological,yin2021solution}. In this paper, however, we focus on the macroscopic Landau--de Gennes (LdG) theory, which describes NLC states and phase transitions with a macroscopic order parameter, $\Q$ tensor \cite{de1993physics}. The LdG theory has been widely used in mathematical studies for confined NLCs, in both two-dimensional (2D) and three-dimensional (3D) confinements \cite{hu2016disclination,muvsevivc2006two,bajc2016mesh}. In a 2D setting, the reduced Landau--de Gennes (rLdG) model, which ignores the ``out-of-plane" order, has been successfully used in a number of NLC studies \cite{han2020reduced,canevari2017order,robinson2017molecular}. This reduced model can capture the type, dimension, location of defects, and the profile of directors in a plane with two degrees of freedom. 

A topologically confined NLC system can admit of multiple stable and unstable equilibria, which corresponds to critical points of NLC free energy. One can classify these critical points by using the Morse index \cite{milnor2016morse}, i.e., the number of negative eigenvalues of their Hessian. The coexistence of multiple states is a desirable feature in the liquid-crystal display industries that produce bistable and multistable liquid crystal devices \cite{kitson2002controllable}. 
In particular, the square domain is a commonly used as a 2D confinement in the study of NLCs \cite{canevari2017order,robinson2017molecular,tsakonas2007multistable,kralj2014order,yin2020construction}. For example, by applying a 2D LdG model, Tsakonas {\it et al.} reported two stable experimentally observed NLC states confined in 2D square \cite{tsakonas2007multistable}.
When the domain size is sufficiently small, the Well Order Reconstruction Solution (WORS) with a pair of orthogonal line defects is the unique solution \cite{kralj2014order}. The WORS exists for all square domains but loses stability when domain size increases. When the domain size is large, two kinds of stable states emerge: a diagonal (D) solution, for which the nematic director aligns along one of the square diagonals, and a rotated (R) solution, for which the director rotates by $\pi$ radians between a pair of opposite square edges \cite{canevari2017order}. In a recent work, Yin {\it et al.} proposed a solution landscape---that is, a pathway map of all connected critical points---to systematically compute possible nematic equilibria confined on a square domain \cite{yin2020construction}.
%Using the WORS as the parent state (the highest-index saddle point), both the solution landscape and the bifurcation diagram of confined NLCs on a square are computed by using the saddle dynamics.

%The solution landscapes of the NLC systems can be greatly affected by the geometry \cite{han2020solution}. 
Although the square domain is well studied in existing literature, it constitutes a special geometry domain for NLC systems. 
For example, the cross structure of the WORS is not generic and will not be found on any other regular polygons \cite{han2020solution}. Moreover, the symmetry anisotropy results in a loss of degeneracy between some rotationally equivalent solutions on the square. 
Thus, symmetry breaking in the square, e.g., the rectangle, may lead to a huge change of the nematic equilibria.
Multiple stable states in the rectangle have been reported in \cite{yao2018topological,fang2020surface}.
For instance, in a nanoscale rectangle, there is a unique solution, BD-S, which has two line defects along the opposite short edges. On the other hand, in a macroscopic rectangle, there are two kinds of R solutions, depending on which of the opposite edges the director rotates \cite{fang2020surface}. What is more, there are different types of bifurcations in the rectangular confinement, while only the pitchfork bifurcations are observed in rLdG studies of the square domain \cite{canevari2017order,robinson2017molecular,yin2020construction}. Different bifurcations have varied influences on nematic equilibria of the rLdG system. The WORS, for example, always gives the parent state (the highest-index saddle point) of the solution landscape on the square \cite{yin2020construction}, whereas some saddle-node bifurcations (i.e., new solutions that appear without connecting to other branches) in the hexagon change the parent state from a Ring solution to a new critical solution \cite{han2020solution}. Therefore, it is very natural to study nematic equilibria inside rectangle domains to study the effect of geometrical anisotropy.

%We classify them by the Morse index \cite{milnor2016morse}, the number of negative eigenvalues of their Hessian. The stable state is the minima of free-energy functional, and it has no negative eigenvalues, i.e., stable state is index-$0$. The transition state between two stable states is maximum along the transition pathway and minimum in any other directions, hence is index-$1$. By applying the high-index optimization-based shrinking dimer method \cite{2019High}, we can find various unstable states with higher indices, and investigate the relationship between stable/unstable states. By combining upward and downward algorithms, we can construct the solution landscape, as \cite{yin2020construction,han2020solution}. Furthermore, we can construct the bifurcation diagrams by tracking the Morse indices of the solutions in solution landscapes.

In this paper, we investigate the solution landscapes and bifurcation diagrams of the rLdG model confined on 2D rectangles with a tangent boundary condition. There are two geometry-dependent variables in a rectangle: the half short edge length $\lambda$ and the aspect ratio $b$. First, we analytically prove that for any $b$ with a small $\lambda$, or for a large $b$ with a fixed domain size, there is a unique stable solution with two line defects on the opposite short edges.
Next, we numerically construct the solution landscapes of the rLdG model by varying the half short edge length $\lambda$ and the aspect ratio $b$. We report a novel X state, which is the analog of the WORS, that emerged from a saddle-node bifurcation and now serves as the parent state of the rectangle's solution landscape. Moreover, various new classes are found in such solution landscapes, including X class, S class, and L class. With a large $\lambda$, X class includes the high-index saddle points such as the X state, and the solutions in S and L classes have multiple interior point defects and line defects along the short and long edges, such as BD-S and BD-L, respectively. By tracking Morse indices of individual solutions, we present bifurcation diagrams to investigate the emergence mechanism and the effect of geometrical anisotropy on nematic equilibria.
%As $b$ increases (from $1.25$ to $1.5$), the number of solutions in L class increases, while the number of solutions in S class remain unchanged. The aspect ratio $b$ also has effects on the type of bifurcation e.g. the bifurcation type of the appearance of R-L, the analog of rotated state in square, changes from pitchfork bifurcation to saddle-node bifurcation as $b$ increases. As $b$ increases further, some high-index solutions in X and L class disappear through saddle-node bifurcation. The aspect ratio of the rectangle also has a great influence on the transition between stable states, e.g., for rectangle with a larger aspect ratio, D state is easier to break through the energy barrier and switches to R-S. 
%, and the aspect ratio $b$apply the reduced LdG model with tangent boundary conditions confined in rectangle to address the following questions--what are the analogs of the reported solutions in the square \cite{robinson2017molecular,yin2020construction}? Since there is the unique solution for small enough $\lambda$, what are the emergence mechanisms or the bifurcation diagram of these solutions by increasing $\lambda$? How the aspect ratio of the rectangle affects the solution landscape and bifurcation diagram? The paper is organized as follows. 

This paper is organized as follows.
In Section 2, we briefly review LdG and rLdG theories. In Section 3, we analytically prove that for any ratio $b \geqslant 1$ with a small enough $\lambda$, or with a large enough $b$ of a fixed domain size, the critical solution of the rLdG free energy is unique.
In Section 4, we describe the numerical method of saddle dynamics and the construction of the solution landscape. 
In Section 5, we systematically study the solution landscapes and bifurcation diagrams as a function of $\lambda$ and $b$. We finally present our conclusion and discussion in Section 6.

\section{The Landau--de Gennes theory}\label{sec:framework}
The LdG theory describes the NLC state with a macroscopic order parameter---the $\Q$-tensor, which is a symmetric, traceless $3\times 3$ matrix. A $\Q$-tensor is said to be isotropic if $\Q = 0$, uniaxial if $\Q$ has a pair of degenerate nonzero eigenvalues, and biaxial if $\Q$ has three distinct eigenvalues \cite{de1993physics,mottram2014introduction}.
A uniaxial $\Q$-tensor can be written compactly as
\begin{equation}\label{eq:uniaxial}
 \Q = s\left(\n \otimes \n - \frac{\mathbf{I}}{3} \right),
\end{equation}
where $\I$ is the identity matrix, $s$ is an order parameter which measures the degree of orientational order and $\n$ is referred to as the director which is the eigenvector corresponding to the non-degenerate eigenvalue. The director labels the single distinguished direction of uniaxial nematic alignment \cite{primicerio1995eg}. 
%the eigenvector corresponding to the single eigenvalue) that denotes the single distinguished direction of orientational order  
\begin{figure}[H]
 \centering
 \includegraphics[width=.45\textwidth]{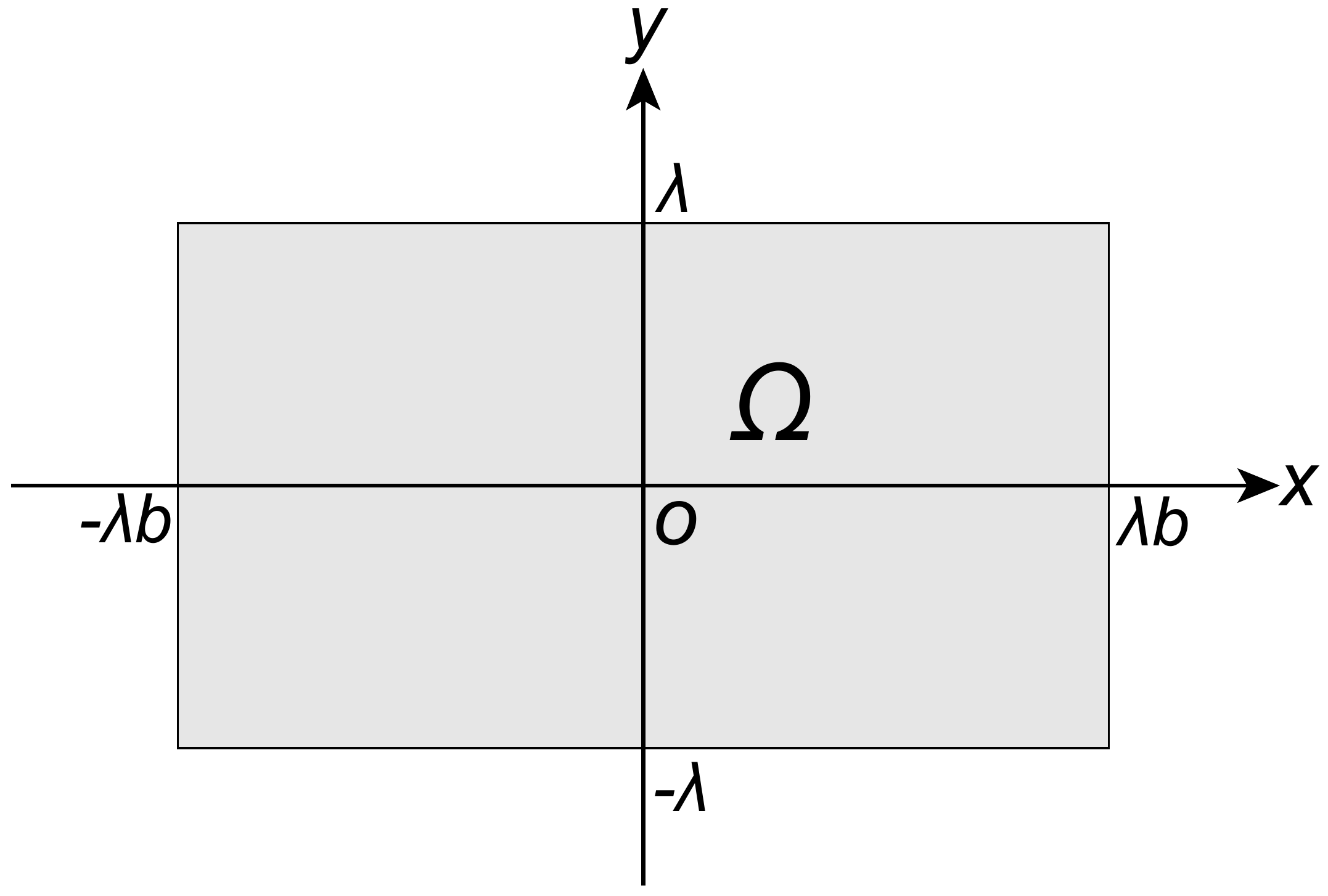}
 \caption{The rectangular domain}
 \label{Omega}
\end{figure}

We work with a simple form of the LdG free energy as
\begin{equation}
 E[\Q]: = \int_{\Omega} \left[\frac{L}{2}\left| \nabla \Q \right|^2 + f_B\left( \Q \right)\right]\mathrm{d}A. 
\end{equation}
The working domain is a rectangle, $\Omega=\left[ -b\lambda,b\lambda\right]\times \left[-\lambda,\lambda \right]$ with the half short edge length $\lambda$ and the aspect ratio $b>1$ (Fig. \ref{Omega}). 
 %$\left|\nabla \Q\right|^2=\dfrac{\partial Q_{ij}}{\partial r_k} \dfrac{\partial Q_{ij}}{\partial r_k}$ is the isotropic elastic energy density/potential, and 
$L$ is a positive material-dependent elastic constant. We work with the simplest form of the elastic energy density--the one-constant elastic energy density, and $f_B$ is a quartic polynomial bulk energy density, i.e.,
\begin{equation}\label{f_B}
	\begin{aligned}
 \left| \nabla \Q \right|^2:&=\frac{\partial Q_{ij}}{\partial r_k}\frac{\partial Q_{ij}}{\partial r_k},\ i,j = 1,\cdots,3,\ k = 1,2,\\
 f_B&(\Q): = \frac{A}{2}\mathrm{tr} \Q^2 - \frac{B}{3} \mathrm{tr} \Q^3 + \frac{C}{4} (\mathrm{tr} \Q^2)^2,
\end{aligned}
\end{equation}
where $\mathbf{r}:=(x,y)$ and we use the Einstein summation convention. 
The thermotropic potential $f_B$ drives the isotropic-nematic phase transition as a function of the temperature.
The variable $A=\alpha (T - T^*)$ is the rescaled temperature, $\alpha>0$ and $T^*$ is a characteristic liquid crystal temperature, below which the isotropic phase $\Q=0$ loses stability. $\alpha, B, C>0$ are material-dependent bulk constants.  

The rLdG model with order parameter $\P$-tensor, a $2\times 2$ traceless and symmetric matrix, has been widely used for the systems confined in the 2D domain both for capturing the qualitative properties of physically relevant solutions and for probing into defect cores \cite{muvsevivc2006two,brodin2010melting, bisht_epl,gupta2005texture,Mu2008Self}. 
From \cite{golovaty2017dimension}, one can restrict the Q-tensors with $\mathbf{z}$ as a fixed eigenvector. 
For a special temperature $A = -\frac{B^2}{3C}$, the order parameter $\mathbf{Q}$ in LdG model can be reduced to $\mathbf{P}$ by removing the ``out-of-plane" order \cite{canevari_majumdar_wang_harris}. The relationship between the LdG and rLdG models can be found in \cite{han2020reduced}. 

The order parameter $\P$ can be written as
\begin{equation}
 \P =
 \left(
 \begin{tabular}{cc}
  $P_{11}$  & $P_{12}$ \\
  $P_{12}$  & $-P_{11}$ 
 \end{tabular}
 \right) = r(\mathbf{m}\otimes\mathbf{m}-\mathbf{I}/2),
 \label{eq: P-tensor} 
\end{equation}
where $\mathbf{m}$ is the in-plane nematic direction, $r$ is the nematic order on $\mathbf{m}$, $\mathbf{I}$ is the $2\times 2$ identity matrix. We track defects by using the nodal set or the zero set of $\P$, which is consistent with disorder in the plane of $\Omega$. 

The rLdG energy is given by
 \begin{equation}\label{p_energy}
  E[\P]: = \int_{\Omega} \left[\dfrac{L}{2}|\nabla \P|^2+ \left(-\dfrac{B^2}{4C}\mathrm{tr}\P^2+\dfrac{C}{4}\left(\mathrm{tr}\P^2\right)^2\right) \right]\mathrm{d} A.
 \end{equation}
%We use MBBA as a representative NLC material and set $B = 0.64\times 10^4 N/m^2$ and $C=0.35\times10^4N/m^2$ throughout this manuscript, which is largely devoted to the study of static equilibria as a function of geometrical parameters \cite{wojtowicz1975introduction,yin2020construction,han2020reduced,han2020solution}.
The physically relevant states are modelled as local or global energy minima subject to the imposed boundary conditions.

%The critical points of the reduced 2D framework also exist in a 3D setting, for example on a well with a 2D rectangle as cross-section \cite{han2020reduced,canevari_majumdar_wang_harris}. However, the Morse index of a 2D saddle point maybe different in the 3D setting. A reduced LdG tensor only has two degrees of freedom while the full LdG tensor in 3D has five degrees of freedom, allowing for
%more instabilities. What's more, there may be other classes of physically relevant solutions in a fully 3D system, such as escaped solutions, with additional degrees of freedom \cite{han2019transition, canevari_majumdar_wang_harris,wang2019order}.
We impose the tangent Dirichlet boundary condition, which requires $\m$ in \eqref{eq: P-tensor} to be tangent to $\partial \Omega$. Such boundary condition means the well molecules in contact with surfaces are constrained to be in the plane of these surfaces, and it is consistent with the experiments \cite{tsakonas2007multistable}. There is a mismatch at four corners, we take the same linear interpolation approach in \cite{han2020reduced,fang2020surface,luo2012multistability}. The Dirichlet condition, $\P=\P_{bc}$ on $\partial \Omega$, which is defined in terms of a function with $0<\epsilon\ll 1/2$,
\begin{equation}
    \label{bc}
    \begin{aligned}
    &\P_{bc}(x=\pm b \lambda,y) =\frac{B}{2C}
    \left(
    \begin{tabular}{cc}
        $-T_\epsilon(\frac{y}{\lambda})$  & $0$ \\
        $0$  & $T_\epsilon(\frac{y}{\lambda})$ 
    \end{tabular}
    \right),\\
    &\P_{bc}(x,y=\pm \lambda) =\frac{B}{2C}
    \left(
    \begin{tabular}{cc}
        $T_\epsilon(\frac{x}{b\lambda})$  & $0$ \\
        $0$  & $-T_\epsilon(\frac{x}{b\lambda})$ 
    \end{tabular}
    \right),
    \end{aligned}
\end{equation} 
where
\begin{equation}
    T_\epsilon(t)=\begin{cases}
        (1+t)/\epsilon, \ -1\leqslant t \leqslant -1+\epsilon,\\
        1, \ |t|\leqslant 1-\epsilon,\\
        (1-t)/\epsilon, \ 1-\epsilon\leqslant t \leqslant 1. 
    \end{cases}
\end{equation}
With sufficiently small $\epsilon>0$, the qualitative solution profiles are not changed by the choice of the interpolation.

We define our admissible space as
\begin{equation}\label{A}
 \mathcal{A}=\left\{ \P\in H^1(\Omega;S_0); \P=\P_{bc} \text{ on } \partial \Omega \right\},
\end{equation}
where
\begin{equation}
 S_0=\left\{\P \in \mathbb{M}^{2 \times 2};\P_{ij}=\P_{ji},\P_{ii}=0\right\}.
\end{equation} 

\section{Theoretical analysis}\label{sec: theoretical analysis}
We rescale the system with $(\bar{x},\bar{y}) = (\frac{x}{\lambda b},\frac{y}{\lambda})$, and define a new parameter ${\bar{\lambda}}^2 = \dfrac{2C\lambda^2}{L}$. With $\bar{E}[\P]=E[\P]/b$, the energy in \eqref{p_energy} is rescaled to
\begin{equation}
	\bar{E}[\P]=\int_{\bar{\Omega}} \left[\dfrac{1}{2b^2}\left| \P_{x} \right|^2 + \dfrac{1}{2}\left| \P_{y} \right|^2+\frac{\bar{\lambda}^2}{2C}\left(-\frac{B^2}{4C}\left|\textbf{P}\right|^2+\frac{C}{4}\left|\textbf{P}\right|^4\right) \right]\mathrm{d}\bar{A}.
	\label{energy_b}
\end{equation}
where $\bar{\Omega}=\left[ -1,1\right]\times \left[-1,1 \right]$ is the rescaled domain, $\mathrm{d}\bar{A}$ is the rescaled area element.
For brevity, we drop all the ``bars" for the rest of this section. $\left|\textbf{P}\right|$ is defined by the inner product of the tensor, $\left<\P_1,\P_2\right>=\mathrm{tr}(\P_1\P_2)$, $\left|\textbf{P}\right|=\sqrt{\left<\P,\P\right>}$. 

The corresponding Euler-Lagrange equations are
\begin{equation}
	\begin{aligned}
	&\mathcal{L}_b P_{11}=\lambda^2  \left(P_{11}^2+P_{12}^2-\frac{B^2}{4C^2}\right)P_{11}, \\
	&\mathcal{L}_b P_{12}=\lambda^2 \left(P_{11}^2+P_{12}^2-\frac{B^2}{4C^2}\right)P_{12},
	\end{aligned}
	\label{EL_b}
\end{equation}
where $\mathcal{L}_b = \frac{1}{b^2} \partial_{xx}+ \partial_{yy}$ is a strongly elliptic differential operator.
%In our framework, X, BD-L, BD-S, WORS and BD have a common feature:$P_{12}=0$  which implies that $\\vec{n}=(1,0)$ or $\\vec{n}=(0,1)$ everywhere. At this time, LdG energy has a simple form:
%\begin{equation}
%	E2_b=\int_\Omega \frac{1}{2}\left |\nabla P_{11}\right|^2+\lambda P_{11}^4-\frac{\lambda B^2}{2C^2}P_{11}^2 \text{dx}
%	\label{haha}
%\end{equation} 
%And the corresponding Euler-Lagrangian equation is:
%\begin{equation}
%	\Delta P_{11}=4\lambda(P_{11}^2-\frac{B^2}{4C^2})P_{11} 
%\end{equation}
%with the boundary condition of $P_{11}$. 

%In \cite{canevari2017order}, the authors prove that there exists a constant $\lambda_0(B, C, \Omega)$, when $\lambda < \lambda_0$ which means with a small domain size, \eqref{energy} has a unique critical point. A question is: under what circumstances, \eqref{energy} also has a unique critical point, when we fix the physical area size but change the geometrical aspect ratio of the rectangle. We prove that $\lambda_0$ is independent of $b$ for the rectangle. Thus, when the physical area size $S=4\lambda^2b$ is fixed and $b$ is large enough which means a thin rectangle, \eqref{energy} has a unique critical point.

\begin{proposition}\label{pro:2}
	For any $b\geqslant 1$ and $\lambda\geqslant 0$, the energy functional in \eqref{energy_b} has a critical point $\P\in \mathcal{A}$,  which satisfies $P_{12}\equiv 0$.
\end{proposition}
\begin{proof}
   This proof is a direct consequence of Proposition 3.2 in \cite{canevari2017order}. The rLdG free energy \eqref{energy_b} has a branch of solutions given by 
   \begin{equation}
	   \P =
	   \left(
	   \begin{tabular}{cc}
		   $P_{11}$  & $0$ \\
		   $ 0$  & $-P_{11}$ 
	   \end{tabular}
	   \right),
	   \label{P12_0}
   \end{equation}
   where $P_{11}$ is defined as minimizer of 
   \begin{equation}
	   H(p)=\int_{\Omega} \frac{1}{b^2}p_{x}^2+p_{y}^2+\frac{\lambda^2}{2C}\left(-\frac{B^2}{2C}p^2+C p^4\right).
	   \label{H_energy}
   \end{equation}
   subject to the Dirichlet conditions \eqref{bc}.
   The existence of the minimizer follows from the direct methods in the calculus of variations, since \eqref{H_energy} is both coercive and weakly lower semi-continuous \cite{canevari2017order}.
   Then $P_{11}$ is a classical solution of the associated Euler-Lagrange equation
   \begin{equation}
	   \mathcal{L}_b p=\lambda^2  (p^2-\frac{B^2}{4C^2})p,
   \end{equation} 
   which ensures that \eqref{P12_0} is a solution of \eqref{EL_b}.
\end{proof}

\begin{proposition}\label{pro:1}
	For any $B, C>0$, $b\geqslant 1$ and $0\leqslant \lambda < \frac{C}{2B}$, the rLdG energy \eqref{energy_b} has a unique critical point $\P \in \mathcal{A}$, which satisfies $P_{12} \equiv 0$.
\end{proposition}
\begin{proof}
The existence of the critical point of \eqref{energy_b} in the admissible space $\mathcal{A}$ is proved in Proposition \ref{pro:2}.

We follow the uniqueness criterion argument in Lemma 8.2 of \cite{lamy2014bifurcation}. For any $B$, $C>0$ and $b\geqslant 1$, if $\P\in \mathcal{A}$ is a critical point of the rLdG energy \eqref{energy_b}, then $\P$ is bounded.
This is an immediate consequence of the maximum principle. By replacing operator $\nabla$ with $\mathcal{L}_b$ and following the calculations in the Lemma B.3. of \cite{lamy2014bifurcation}, we have $|\P|^2 \leqslant \frac{B^2}{2C^2}$. We define the convex set $\mathcal{S}=\{\P\in \mathcal{A}, |\P|^2\leqslant \frac{B^2}{2C^2}\}$. 

Then, we prove that $E$ is strictly convex on $\mathcal{S}$.
For any $\P, \bar{\P} \in \mathcal{S}$, we have
\begin{equation}
	\begin{aligned}
	&E(\frac{\P+\bar{\P}}{2})-\frac{1}{2}E(\P)-\frac{1}{2}E(\bar{\P}) \\
	&= \int_{\Omega} - \frac{1}{8b^2}|(\P-\bar{\P})_x|^2-\frac{1}{8}|(\P-\bar{\P})_y|^2 \mathrm{d}A + \int_{\Omega} f(\frac{\bar{\P}+\P}{2})-\frac{1}{2}f(\P)-\frac{1}{2}f(\bar{\P})\mathrm{d}A.
	\end{aligned}
\label{E_convex}
\end{equation}
where $f(\P)$ is the bulk energy density of \eqref{energy_b}, i.e.,
\begin{equation}
	f(\P)=\frac{\lambda^2}{2C}\left(-\frac{B^2}{4C}\left|\P\right|^2+\frac{C}{4}\left|\P\right|^4\right).
\end{equation}

For any point $(\hat{x},\hat{y})\in \Omega$, we have
\begin{equation}
	 (P_{1i}-\bar{P}_{1i}) (\hat{x},\hat{y}) =\int_{-1}^{\hat{y}}(P_{1i}-\bar{P}_{1i})_y(\hat{x},y)\mathrm{d}y.
\end{equation}
By using Schwarz inequation, we have
\begin{equation}
	 (P_{1i}-\bar{P}_{1i})^2(\hat{x},\hat{y})  \leqslant |\hat{y}+1| \int_{-1}^{\hat{y}}{(P_{1i}-\bar{P}_{1i})}_y ^2(\hat{x},y) \mathrm{d}y \leqslant 2 \int_{-1}^{1} (P_{1i}-\bar{P}_{1i})_y ^2(\hat{x},y) \mathrm{d}y.
\end{equation}
Integrating the both sides on $\Omega$, we have
\begin{equation}
	\int_{\Omega} (P_{1i}-\bar{P}_{1i}) ^2(\hat{x},\hat{y}) \mathrm{d}\hat{x}\mathrm{d}\hat{y}\leqslant 4\int_{\Omega} (P_{1i}-\bar{P}_{1i})_y^2(x,y) \mathrm{d}x\mathrm{d}y,
	\label{eq: poincare_pre}
\end{equation}
i.e., the Poincare inequality
\begin{equation}
	\Vert \P-\bar{\P}\Vert_{L_{\Omega}^2}^2\leqslant 4\Vert (\P-\bar{\P})_y\Vert_{L_{\Omega}^2}^2,
	\label{eq: poincare}
\end{equation}
where we define the $L^2$-norm as $\Vert \P\Vert_{L_{\Omega}^2}= \left({\int_{\Omega} |\P|^2}\text{d}A\right)^{\frac{1}{2}}$. The rationality of exchanging the order of integration in \eqref{eq: poincare_pre} can be obtained from the density of $C_0^\infty(\Omega)$ in $H_0^1(\Omega)$, i.e., we can assume $P_{1i}-\bar{P}_{1i}\in C_0^\infty(\Omega)$.

%We estimate the bound of the second integral in \eqref{E_convex}. Since $|\P|^2|\bar{\P}|^2-\left<\P,\bar{\P}\right>^2\geqslant 0$, we have
%\begin{equation}
%\begin{aligned}
%&-\frac{1}{8}\left(\left|\frac{\P+\bar{\P}}{2}\right|^4-\frac{1}{2}|\P|^4-\frac{1}{2}|\bar{\P}|^4\right) \\
%&= \frac{7(|\P|^2+|\bar{\P}|^2)|\P-\bar{\P}|^2+10\left<\P,\bar{\P}\right>|\P-\bar{\P}|^2-16(|\P|^2|\bar{\P}|^2-\left<\P,\bar{\P}\right>^2)}{128},\\
%&\leqslant \frac{7(|\P|^2+|\bar{\P}|^2)|\P-\bar{\P}|^2+10\left<\P,\bar{\P}\right>|\P-\bar{\P}|^2}{128}.
%\end{aligned}
%\end{equation}

We estimate the bound of the second integral in \eqref{E_convex}. Since $|\P|^2|\bar{\P}|^2-\left<\P,\bar{\P}\right>^2\geqslant 0$, $|\P|^4$ is a convex function on $\mathcal{S}$ and $|\P|^2,|\bar{\P}|^2,\left<\P,\bar{\P}\right>\leqslant \frac{B^2}{2C^2}$, we have
%	-\frac{7(|\P|^2+|\bar{\P}|^2)|\P-\bar{\P}|^2+10\left<\P,\bar{\P}\right>|\P-\bar{\P}|^2-16(|\P|^2|\bar{\P}|^2-\left<\P,\bar{\P}\right>^2)}{128}\right|\\
%	&\leqslant \lambda^2\left( \frac{B^2}{32C^2}|\P-\bar{\P}|^2+\left|\frac{7(|\P|^2+|\bar{\P}|^2)|\P-\bar{\P}|^2+10\left<\P,\bar{\P}\right>|\P-\bar{\P}|^2-16(|\P|^2|\bar{\P}|^2-\left<\P,\bar{\P}\right>^2)}{128}\right|\right).
%Since $|\P|^4$ is a convex function of $\P \in \mathcal{S}$, 
%\begin{equation}
%	\begin{aligned}
%	&\frac{7(|\P|^2+|\bar{\P}|^2)|\P-\bar{\P}|^2+10\left<\P,\bar{\P}\right>|\P-\bar{\P}|^2-16(|\P|^2|\bar{\P}|^2-\left<\P,\bar{\P}\right>^2)}{128}\\
%	&=-\frac{1}{8}\left(\left|\frac{\P+\bar{\P}}{2}\right|^4-\frac{|\P|^4+|\bar{\P}|^4}{2}\right)\geqslant 0.
%\end{aligned}
%\end{equation}
%We have used the fact that $|\P|^4$ is a convex function on $\mathcal{S}$.
%Note that $|\P|^2,|\bar{\P}|^2,\left<\P,\bar{\P}\right>\leqslant \frac{B^2}{2C^2}$,
\begin{equation}
\left|f(\frac{\P+\bar{\P}}{2})-\frac{1}{2}f(\P)-\frac{1}{2}f(\bar{\P})\right|\leqslant \frac{B^2\lambda^2}{8C^2}|\P-\bar{\P}|^2.
\end{equation}
By using the Poincare inequality in \eqref{eq: poincare}, we have
\begin{equation}
||f(\frac{\P+\bar{\P}}{2})-\frac{1}{2}f(\P)-\frac{1}{2}f(\bar{\P})||_{L_{\Omega}^2}\leqslant \frac{B^2\lambda^2}{2C^2}\Vert (\P-\bar{\P})_y\Vert_{L_{\Omega}^2}^2.
\end{equation}

Thus, for $\lambda<\lambda_0:=\frac{C}{2B}$, the energy functional $E$ in \eqref{energy_b} is strictly convex on $\mathcal{S}$ and has a unique critical point,
since $\forall \P,\ \bar{\P} \in \mathcal{S}$, and $\P \neq \bar{\P}$,
\begin{equation}
	E(\frac{\P+\bar{\P}}{2})-\frac{1}{2}E(\P)-\frac{1}{2}E(\bar{\P})\leqslant - \frac{1}{8}\Vert (\P-\bar{\P})_y\Vert_{L_{\Omega}^2}^2+ \frac{B^2\lambda^2}{2C^2}\Vert (\P-\bar{\P})_y\Vert_{L_{\Omega}^2}^2 < 0.
\end{equation}
\end{proof}

\begin{corollary}
	If the domain size $S=4\lambda^2 b$ is fixed as a constant, there exists a $b_0(B,C)=\max\left(1,\frac{SB^2}{C^2}\right)$ such that, for any $b > b_0$, i.e., a long rectangular domain, \eqref{energy_b} has a unique critical point $\P\in \mathcal{A}$, which satisfies $P_{12}\equiv 0$. 
\end{corollary}

This corollary can be proved by Proposition \ref{pro:1} with $\lambda=\sqrt{\frac{S}{4b}}$. 

\begin{proposition}
Let $\P\in \mathcal{A}$ be a solution of rLdG Euler-Lagrange equations \eqref{EL_b}, for $b>1$ and $\lambda>0$, subject to the boundary condition \eqref{bc}, then $\P$ uniformly converges to the unique solution $\P_0$ of  
 \begin{equation}\label{eq:zero}
	\begin{aligned}
	\mathcal{L}_b (P_0)_{11}=0,\ on\ \Omega, \\
	\mathcal{L}_b (P_0)_{12}=0,\ on\ \Omega,
    \end{aligned}
\end{equation}
subject to the same boundary condition, as $\lambda\to 0$ or $b\to\infty$ with the fixed domain size $S$. The error estimate is 
 \begin{equation}
	 \Vert \P-\P_0\Vert_{L^\infty_{\Omega}}\leqslant \hat{C}\lambda^2,
	 \label{estimate}
 \end{equation}
 for a positive constant $\hat{C}(B,C)$ independent of $\lambda$ and $b$.
 The $L^\infty$-norm is defined as $\Vert \P\Vert_{L_{\Omega}^\infty}=\text{ess sup}_{\r\in \Omega}\left|\P(r)\right|$.

\end{proposition} 
\begin{proof}
By Proposition 13 in \cite{majumdar2010landau}, we have that $\P,\P_0 \in C^\infty(\Omega;S_0)$. Let $v_i=P_{1i}-(P_0)_{1i}$. Due to the bound of $|\P|$ in Proposition \ref{pro:1} and equation \eqref{EL_b}, we have
 \begin{equation}
	 |\mathcal{L}_b v_i| \leqslant \hat{C}_2(B,C) \lambda^2\ on\ \Omega ,\ v_i = 0\ on\ \partial \Omega,\ i=1,2.
 \end{equation}
We set an auxiliary function
\begin{equation}
	\tilde{v}_i=v_i-(e^2-e^{y+1})\hat{C}_2(B,C)\lambda^2,\ i=1,2,
\end{equation}
which satisfies
\begin{equation}
	\mathcal{L}_b \tilde{v}_i=\mathcal{L}_b v_i+e^{y+1}\hat{C}_2(B,C)\lambda^2  \geqslant 0 \ on\ \Omega,\ \tilde{v}_i \leqslant 0 \ on\ \partial\Omega,\ i=1,2.\\
\end{equation}
From the weak maximum principle, $(e^2-e^{y+1})\hat{C}_2(B,C)\lambda^2$ is the super-solution for $v_i, i=1,2$, thus, $v_i\leqslant (e^2-1)\hat{C}_2(B,C)\lambda^2, i=1,2$. Similarly, we have $-v_i\leqslant (e^2-1)\hat{C}_2(B,C)\lambda^2, i=1,2$.
By taking $\hat{C}(B,C)=2(e^2-1)\hat{C}_2(B,C)$, we can get the error estimate in \eqref{estimate}.
\end{proof}

\section{Numerical method} \label{Sec Numerical method}
In order to better visualize the rectangle, we apply the following rLdG energy which can be obtained with $\bar{x}=bx$ in \eqref{energy_b}.

\begin{equation}
	E(\P)=\int_\Omega \frac{1}{2}\left |\nabla \textbf{P}\right|^2+\frac{\lambda^2}{2C}\left(-\frac{B^2}{4C}\left|\textbf{P}\right|^2+\frac{C}{4}\left|\textbf{P}\right|^4\right)\text{dx},
	\label{energy_final}
\end{equation}
where $\Omega=\left[ -b,b \right]\times \left[-1,1 \right]$.

\subsection{Saddle dynamics method}
In order to construct the solution landscape and bifurcation diagram, the saddle dynamics (SD) method is designed to search unstable saddle points with a given index \cite{2019High}. Here, we explain the essential steps in the SD method to find an index-$k$ saddle point for the rLdG energy \eqref{energy_final}. A non-degenerate index-$k$ saddle point $\hat{\P}=(\hat{P}_{11},\hat{P}_{12})$ has the following property. The Hessian $\nabla ^2 E(\hat{\P})$ has exact $k$ negative eigenvalues $\lambda_1 \leqslant \cdots \leqslant \lambda_k$, corresponding to $k$ unit eigenvectors $\hat{\_v}_1,\cdots,\hat{\_v}_k$ satisfying $\big\langle{\hat{\_v}_i}, \hat{{\_v}}_j \big\rangle = \delta_{ij}$, $1\leqslant i, j \leqslant k$. By setting the $k$-dimensional subspace $\mathcal{V}=\text{span} \big \{ \hat{\_v}_1,\cdots,\hat{\_v}_k \big \}$, $\hat{\P}$ is a local maximum on $\hat{\P}+\mathcal{V}$ and a local minimum on $\hat{\P}+\mathcal{V}^\perp$, where $\mathcal{V}^\perp$ is the orthogonal complement space of $\mathcal{V}$. 

The SD dynamics for an index-$k$ saddle point ($k$-SD) is given by,
\begin{equation}
    \left\{
    \begin{aligned}
		\dot{\P}&=- (\I-2\sum_{i=1}^k {\_v}_i{\_v}_i^T)\nabla E(\P), \\
		  \dot{\_v}_i&=-   (\I-{\_v}_i{\_v}_i^T-\sum_{i=1}^{i-1}2{\_v}_j{\_v}_j^T)\nabla^2 E(\P) \_v_i,\ i=1,2,\cdots,k ,\\
    \end{aligned}
    \right.
	\label{eq: SD}
\end{equation}
where $\I$ is the identity operator. For the rLdG energy \eqref{energy_final},
\begin{equation}
	\nabla E(\P)=
	2\left(
	\begin{tabular}{c}
		$-\Delta P_{11}+\lambda^2 (P_{11}^2+P_{12}^2-\frac{B^2}{4C^2})P_{11}$   \\
		$-\Delta P_{12}+\lambda^2 (P_{11}^2+P_{12}^2-\frac{B^2}{4C^2})P_{12}$   
	\end{tabular}
	\right).
\end{equation}
To avoid evaluating the Hessian of $E(\P)$, the dimer
\begin{equation}
    h(\P,\_v_i)=-\frac{\nabla E(\P-l\_v_i)-\nabla E(\P+l\_v_i)}{2l}
\end{equation}
is an approximation of $\nabla ^2 E(\P)\_v_i$, with a small dimer length $2l$. The dynamics for $\P$ in \eqref{eq: SD} can be written as
\begin{equation}
	\begin{aligned}
	\dot{\P}&=\left(\I-\sum_{i=1}^k {\_v}_i{\_v}_i^T\right) \left(-\nabla E(\P)\right)+ \left(\sum_{i=1}^k {\_v}_i{\_v}_i^T\right) \nabla E(\P) \\
	&= \left( \I-\mathcal{P}_{\mathcal{V}}  \right)\left(-\nabla E(\P)\right)+ \mathcal{P}_{\mathcal{V}} \left(\nabla E(\P)\right),
\end{aligned}
\end{equation}
where $\mathcal{P}_{\mathcal{V}}\nabla E(\P)=\left(\sum_{i=1}^k {\_v}_i{\_v}_i^T\right)\nabla E(\P)$ is the orthogonal projection of $\nabla E(\P)$ on $\mathcal{V}$. Thus, $\left( \I-\mathcal{P}_{\mathcal{V}}  \right)\left(-\nabla E(\P)\right)$ is a descent direction on $\mathcal{V}^\perp$, and $\mathcal{P}_{\mathcal{V}} \left(\nabla E(\P)\right)$ is an ascent direction on $\mathcal{V}$. 

The dynamics for $\_v_i, i=1,2,\cdots,k$ in \eqref{eq: SD} can be obtained by minimizing the $k$ Rayleigh quotients simultaneously with the gradient type dynamics,
\begin{equation}
	\min_{{\_v}_i}\text{  }\left<\_v_i, \nabla ^2 E(\P)\_v_i\right>,\ \text{s.t.}\ \left<\_v_i,\_v_j\right>=\delta_{ij}, \ j=1,2,\cdots,i,
\end{equation}
which renews the subspace $\mathcal{V}$ by finding the eigenvectors corresponding to the smallest $k$ eigenvalues of $\nabla^2 E (\P)$.

If the linear steady state of \eqref{eq: SD} is $(\P^*,\_v_1^*,\cdots,\_v_k^*)$, then $\P^*$ is a $k$-index saddle point of $E(\P)$, and $\_v_i^*, i=1,\cdots, k$, is an eigenvector of $\nabla^2 E(\P^*)$ corresponding to the $i$-th smallest eigenvalue \cite{2019High}. A stable state $\hat{\P}$ is a critical point of $E(\P)$ and the smallest eigenvalue of $\nabla^2 E(\hat{\P})$ is positive. 

We use finite difference methods to estimate the spatial derivation by taking the nodes $(x_i,y_j),i=0,1,\cdots,N_1,j=0,1,\cdots,N_2$ with the step length $h=\frac{1}{50}$, where
\begin{equation}
	\begin{aligned}
		-1=x_0 \leqslant x_1\leqslant \cdots\leqslant x_{N_1-1}\leqslant x_{N_1}=1,\ x_i=-1+ih,\\
		-b=y_0\leqslant y_1\leqslant \cdots\leqslant y_{N_2-1}\leqslant y_{N_2}=b,\ y_j=-b+jh.
	\end{aligned}
\end{equation}

\subsection{Construction of the solution landscape and bifurcation diagram}
Following the SD dynamics \eqref{eq: SD}, we construct the solution landscape by two algorithms: the downward search that enables us to search for connected lower-index saddle points; the upward search to find the higher-index saddle points \cite{yin2020construction}.

First, the downward search is used to search possible lower-index saddle points starting from an existing parent state (the highest-index saddle point) by following its unstable directions. We assume that an index-$k$ saddle $\P$ and eigenvectors $\_v_i,i=1,2,\cdots,k$ corresponding to the $i$-th smallest eigenvalue of the Hessian at $\P$ are provided. To search for a lower index-$m$ ($m<k$) saddle point, we choose an unstable direction which is a linear combination of  $\_v_i,i=1,2,\cdots,k$ as the moving direction of the initial state, and $m$ other unstable eigenvectors as the initial directions of $m$-SD. A typical choice of initial condition for a downward search following an $m$-SD is $(\P \pm \epsilon \_v_{m+1},\_v_1,\_v_2,\cdots,\_v_m)$. The small driving force $\pm \epsilon \_v_{m+1}$ keep the system away from the index-$k$ saddle point $\P$. Normally, we can find a pair of index-$m$ saddle points, which corresponding to the positive and negative driving force. If a new parent state emerges or multiple parent states exist, then the upward search is used to find possible new parent state from a lower-index saddle point. We assume that an index-$m$ ($m<k$) saddle $\P$ and eigenvectors $\_v_i,i=1,2,\cdots,k$ corresponding to the $i$-th smallest eigenvalue of the Hessian at $\P$ are provided. To search for an index-$k$ saddle point, we choose a direction which is a linear combination of $k-m$ eigenvectors $\_v_i,i=m+1,\cdots,k$ as the moving direction of the initial state, and $\_v_i,i=1,\cdots.,k$ as initial eigendirections of $k$-SD. A typical choice of initial condition for a upward search following a $k$-SD is $(\P \pm \epsilon \_v_{m+1},\_v_1,\_v_2,\cdots,\_v_k)$.

By repeating the downward search and upward search, we systematically find all possible critical points, including both unstable saddle points and stable minima, and their connections. In the next section, we explore the solution landscapes with various values of $b$ and $\lambda$. The efficiency of this method is reflected in the fact that we find two new branches which are disconnected to the previously reported branches \cite{fang2020surface}. Furthermore, we can also construct the bifurcation diagrams by tracking the indices of the solutions in the solution landscapes. A change of the Morse index is a sign of the bifurcation and possible change of stability properties.

\section{Numerical results} 
\subsection{Typical solutions}

\begin{figure}[hbtp]
    \centering
    \includegraphics[width=\textwidth]{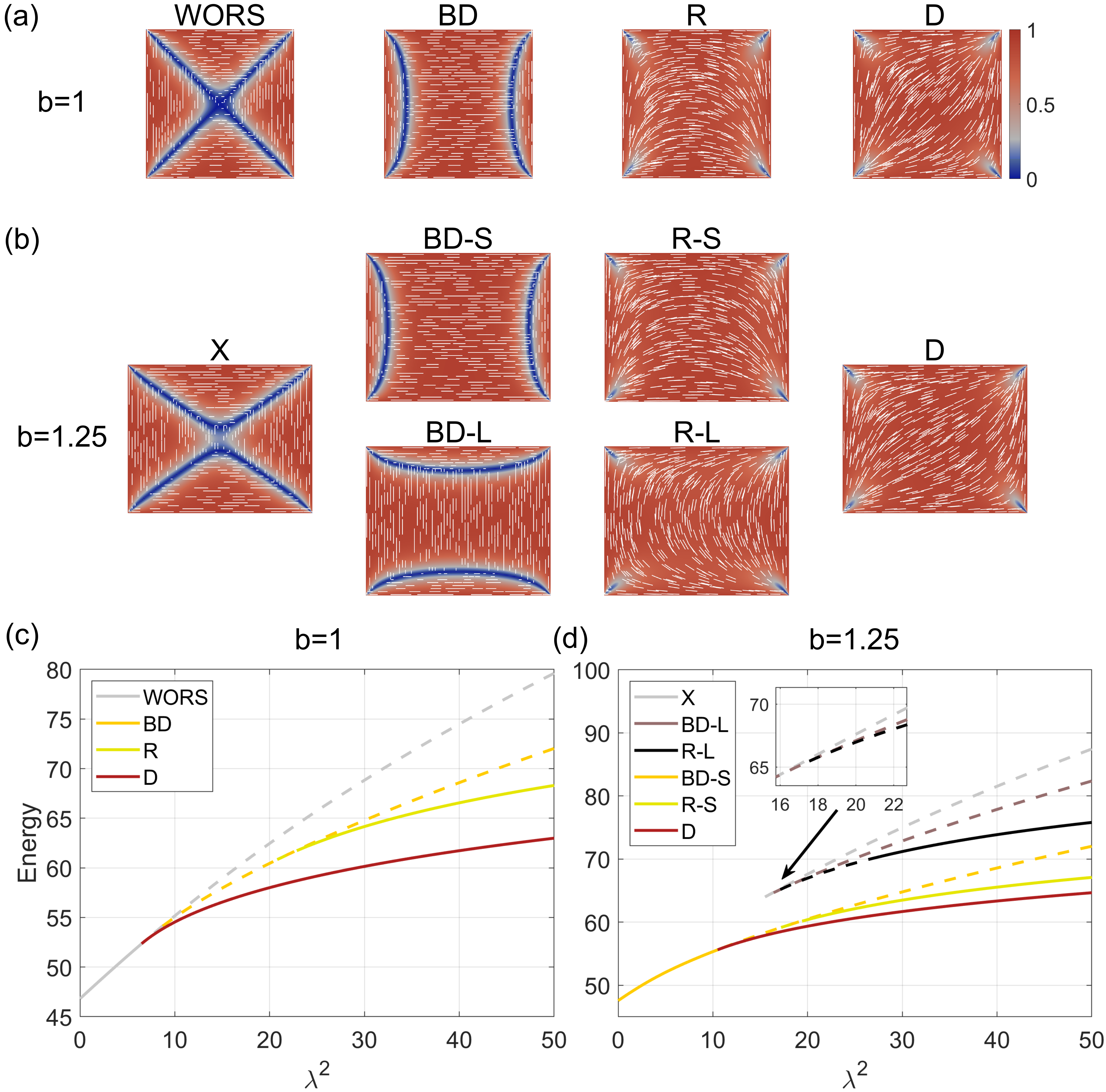}
	\caption{(a) Typical solutions on a square with $b = 1$ and $\lambda^2 = 50$. (b) Typical solutions on a rectangle with $b = 1.25$ and $\lambda^2 = 50$. (c-d) The plots of the rLdG energy \eqref{energy_final} versus $\lambda^2$ for some typical solutions with $b=1$ and $b=1.25$. The color denotes the nematic order $\left|\P\right|/\sqrt{2}$, and the white line encodes the direction of nematic defined in \eqref{eq: P-tensor}. We use the same visualization method for nematic configurations and the same color bar for nematic order in the following figures. The solid (dashed) line denotes a stable (unstable) critical solution.}
	\label{fig:typical}
\end{figure}

Several typical solutions on squares, such as WORS, BD, R, and D (Fig. \ref{fig:typical}(a)), which have been reported in some of the literature \cite{canevari2017order,robinson2017molecular,yin2020construction}. The WORS features a pair of mutually orthogonal line defects on the diagonals. It is a unique (stable) solution in a rLdG model for a small enough square, but loses its stability as the square edge length increases \cite{canevari2017order}. The Morse index of WORS increases as the domain size increases and is higher than other critical points. Thus, WORS is always chosen as a parent state to construct the solution landscape \cite{yin2020construction}. The BD solution is an unstable saddle point and features two symmetric line defects near a pair of opposite edges that partition the square into three regions. The nematic director of BD aligns along horizontally or vertically. When the domain size is large, the rotated (R) solution, for which the director is rotated by $\pi$ radians between a pair of opposite edges, and the diagonal (D) solution, for which the director is along a diagonal of the square, are both stable \cite{tsakonas2007multistable}.

Like these typical solutions on a square, we compute corresponding solutions on a rectangular domain in Fig. \ref{fig:typical}(b). On a rectangle, the parent state that is analogous to WORS is the X solution for a large $\lambda$. The X has line defects near the center region of rectangle but loses the cross structure of the WORS. Unlike the WORS, X exists only for a large enough $\lambda$ and is always unstable. The symmetry breaking of the rectangle brings about a loss of degeneracy between BD-L and BD-S, which are rotationally equivalent on the square. 
%In the following article, we take symmetry into account. 
The BD-S solution, which features a pair of line defects along the two shorter edges of a rectangle, is a unique stable solution for small $\lambda$ and loses its stability as $\lambda$ increases. This is consistent with our theoretical analysis, i.e., BD-S is the unique critical point in Proposition \ref{pro:1} which satisfies $P_{12} \equiv 0$. It is notable that the BD solution is always unstable on a square domain \cite{yin2020construction}. Here, our analysis and numerical results show that, in a fixed rectangular domain size, we can stabilize the BD-S solution by increasing the aspect ratio $b$. This reveals how geometrical anisotropy can affect the stability of nematic equilibria. The BD-L solution, which features a pair of line defects along the two longer edges of the rectangle, is unstable and exists only for a large $\lambda$. We can distinguish between X and BD-L by comparing the distance between the line defects and the long side edges. Similarly, loss of geometrical symmetry brings two ``R-like" solutions, R-S and R-L, for which the nematic director is rotated by $\pi$ radians between two opposed short and long edges respectively. They are stable for large $\lambda$. The analogous D state on rectangle is still the global minimizer of (\ref{energy_final}) for large $\lambda$. However, the director of a D state on the diagonals aligns not strictly along one of the diagonals of the rectangle domain. The stable states on the rectangle, D, R-S, and R-L, have been reported in \cite{yao2018topological,fang2020surface,lewis2014colloidal}.

Next, we track the branches of these typical solutions on a square and investigate the bifurcation of them in Fig. \ref{fig:typical}(c). For small $\lambda^2$, the WORS is a unique solution and its smallest eigenvalue strictly decreases as the domain size \cite{canevari2017order}. As $\lambda^2$ increases, the WORS becomes unstable and bifurcates into an unstable WORS and a stable D state. As $\lambda^2$ continues to increase, the WORS further bifurcates into a new unstable BD branch. The R solutions are bifurcated from the BD and stabilized via a pitchfork bifurcation from index-1 saddles to minima.

An analogous bifurcation of typical solutions on a rectangle is shown in Fig. \ref{fig:typical}(d).
The BD-S is a unique solution on a rectangle when $\lambda^2$ is small. As $\lambda^2$ increases, BD-S becomes unstable and bifurcates into stable D solutions via pitchfork bifurcation. Then this index-1 saddle BD-S further bifurcates into unstable R-S solution, and the R-S solution gains stability with a large $\lambda^2$. The branches of unstable BD-L and X are unconnected to the BD-S branches and appear simultaneously with the same energy. The BD-L solution bifurcates into the R-L solution, and the R-L gains stability with a large $\lambda^2$. It is noteworthy that the energy of the BD-L solution with long line defects is much higher than that of the BD-S solution with short line defects and that the energy gap separates the six typical solutions into two solution families.

\subsection{Solution landscape at $b=1.25$}

\begin{figure}[hbtp]
    \centering
    \includegraphics[width=\textwidth]{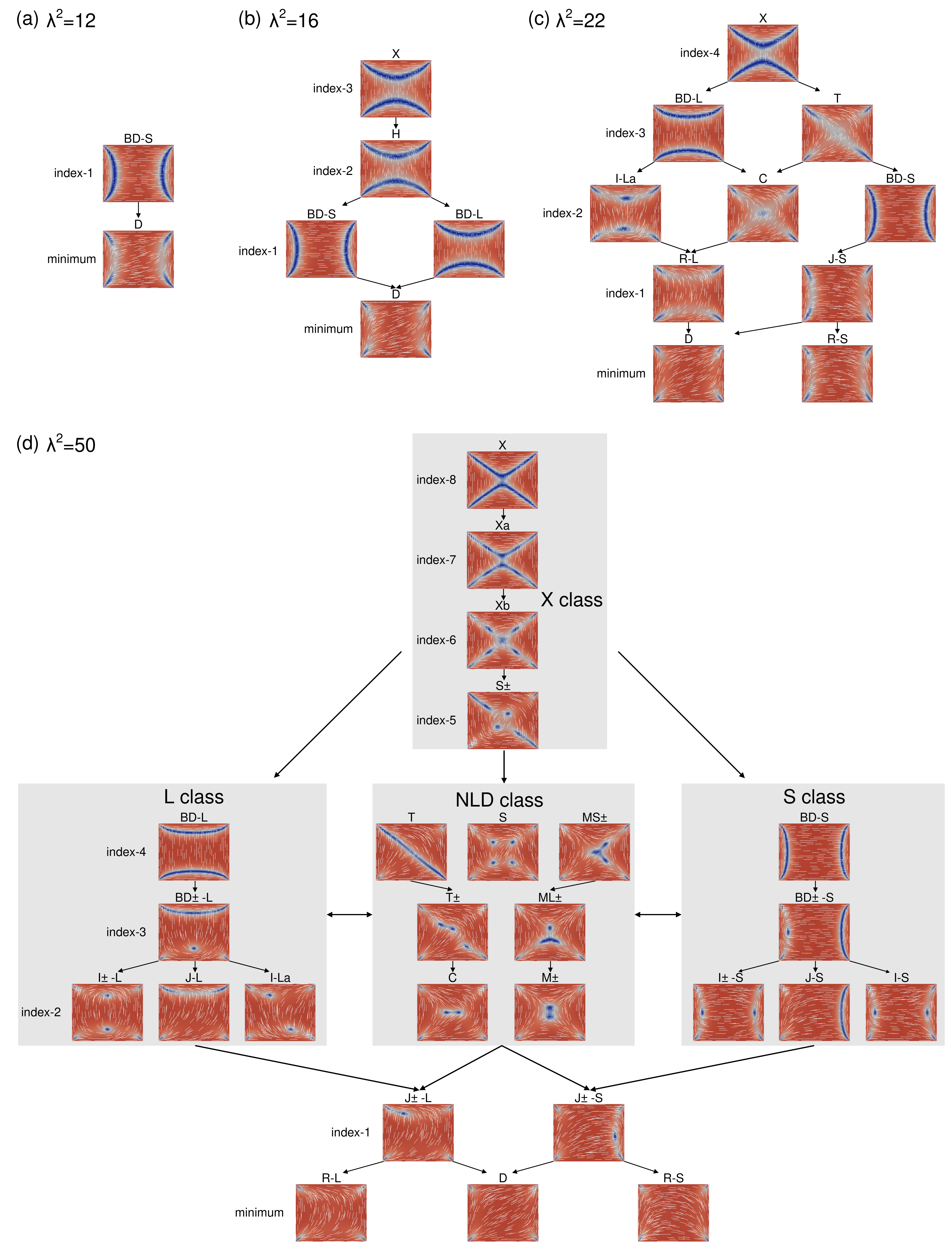}
    \caption{The solution landscapes at $b=1.25$ with (a) $\lambda^2=12$, (b) $\lambda^2=15.7$, (c) $\lambda^2=22$ and (d) $\lambda^2=50$. The arrow from the higher index saddle to the lower index saddle represents that the solution (in the arrow's tail) with a small perturbation that is able to converge to the solution in the arrow's head by following \eqref{eq: SD}. In the following solution landscapes, we represent the connections in the same way.}
	\label{b1.25}
\end{figure}

The solution landscapes for rectangles with a fixed aspect ratio $b=1.25$ are shown in Fig. \ref{b1.25}. At $\lambda^2=12$, only index-$1$ BD-S connects with index-$0$ D state (Fig. \ref{b1.25}(a)). At $\lambda^2=16$, as shown in Fig. \ref{b1.25}(b), index-$3$ X is the new parent state that connects with index-$2$ H. This index-2 H state features a pair of interior point defects, which are produced by interrupting the line defects of X. Following $1$-SD in \eqref{eq: SD} along unstable directions, H converges to two index-$1$ solutions BD-S and BD-L. Hence, there are two paths from the parent state X to the stable state D, which pass though BD-S and BD-L respectively. At $\lambda^2=22$, the indices of X and BD-L increase from index-3 to index-4 and from index-1 to index-3, respectively (Fig. \ref{b1.25}(c)). The index-$4$ X connects with two index-$3$ solutions BD-L and T, and further connects with three index-$2$ solutions I-La, C and BD-S. The C state features a central point defect for small $\lambda^2$, and this point defect splits into two point defects as $\lambda^2$ increases. The director of the C state is similar with R-L, i.e., rotated by $\pi$ radians between two long edges. Following $1$-SD along the unstable direction, I-La and C converge to an index-$1$ R-L; BD-S converges to an index-$1$ J-S. Finally, we can find two stable states, D and R-S, from index-$1$ R-L and J-S.

At $\lambda^2=50$, we have $26$ solutions with various indices, and the connections between them is shown in Fig \ref{b1.25}(d). The parent state is still the X state. R-L, D, and R-S are stable states. We classify the rest of the solutions into four classes: S class, L class, X class and NLD (No Line Defect) class, according to the location of the defects and the connections between them. Solutions in S class have defects near the short rectangular edges. The BD-S is the parent state of S class. Following the unstable eigendirection of the BD-S, one of the line defects splits into two $\pm 1/4$ point defects near the corners (the director defined in \eqref{eq: P-tensor} rotates by $\pm \pi/2$ radians anticlockwise around the defect core) and one $\pm 1/2$ point defect near the middle of the short edge and the 3-SD converges to index-$3$ BD$\pm$-S (the $\pm$ signs represent $\pm 1/2$ point defects). Following the unstable eigendirection of 
the BD$\pm$-S, another line defect is interrupted as a $\pm 1/2$ point defect (I$\pm$-S) or a $\mp 1/2$ point defect inside (I-S). J-S can be obtained by moving the $\pm 1/2$ point defect along the short edge and merging with the $\pm 1/4$ point defects near one of the corners. The structure of L class in the solution landscape is analogous to that of the S class. The BD-L with long line defects is the parent state.  We can obtain BD$\pm$-L when one of the BD-L's line defects splits into point defects, further, through the splitting of the other line defect, we can obtain I$\pm$-L and I-La. The two point defects in I-La deviate from the middle of the long edges and get closer to opposed short edges.

The solutions in X class locate the defects at the position of the line defects of X. In X class, the pairs of line defects of index-8 X splits into multiple point defects. Index-$7$ Xa, index-6 Xb, index-5 S$\pm$ inset with a tiny I-L, C, and D profile in the center, respectively. We classify all other solutions that have no preference for the position of defects as NLD class. X class can connect with L class, NLD class, and S class. The connection between L class and S class can be achieved through NLD class or X class. Hence, there are four paths from the parent state X to the stable state D: X class$\rightarrow$L class$\rightarrow$J$\pm$-L$\rightarrow$D, X class$\rightarrow$S class$\rightarrow$J$\pm$-S$\rightarrow$D, X class$\rightarrow$NLD class$\rightarrow$J$\pm$-L$\rightarrow$D, and X class$\rightarrow$NLD class$\rightarrow$J$\pm$-S$\rightarrow$D. Similarly, we have two paths from X to R-L(S) through L(S) class or NLD class. 

\begin{figure}[hbtp]
    \centering
    \includegraphics[width=.99\textwidth]{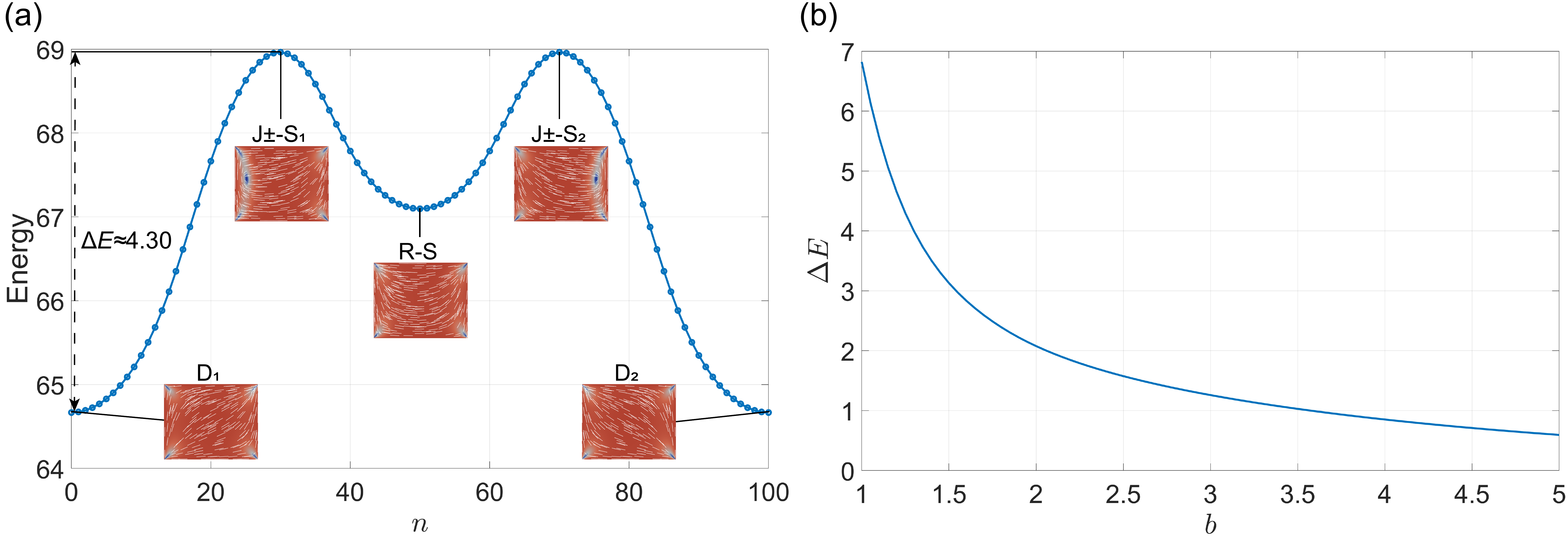}
    \caption{(a) Transition pathway between D$_1$ and D$_2$ on a rectangle with $b=1.25$ and $\lambda^2=50$ (b) For a fixed area size $S=250$, the energy barrier ($\Delta E=E_{\text{J}\pm \text{S}}-E_{\text{D}}$) from D to R-S versus the aspect ratio $b$.}
	\label{MEP}
\end{figure}

Under thermal fluctuations or external disturbances, the NLC system can transform from a metastable state to another one, leading to a sharp change in the location and topology of the defect pattern \cite{han2019transition,kusumaatmaja2015free}. The transition state is the index-1 saddle point of \eqref{energy_final} and plays a key role in determining the energy barrier of such a transition pathway. The index-$1$ J$\pm$-S(L) is the transition state between R-S(L) and D, and the transition pathway between D and R-S(L) is D$\leftrightarrow$J$\pm$-S(L)$\leftrightarrow$R-S(L). The transition pathway (i.e. the minimum energy path) between two stable D solutions, D$_1$ and D$_2$, is plotted in (Fig. \ref{MEP}(a)). It shows that D$_1$ and D$_2$ cannot be connected by a single transition state, and the transition pathway between D$_1$ and D$_2$ follows a two-stage phase transition that involves the metastable R-S and overcomes two energy barriers brought on by J$\pm$S$_1$ and J$\pm$S$_2$. 
The rectangle's aspect ratio greatly influences this process. For a fixed rectangular domain size, a higher aspect ratio leads to a lower energy barrier (Fig. \ref{MEP}(b)). This result indicates that larger geometrical anisotropy is more advantageous for switching between two D states because of its reduced energy barrier, a result which may have practical significance for designing bistable liquid crystal devices.

\subsection{Bifurcation diagram as a function of $\lambda^2$ at $b=1.25$}

By tracking the indices of the solutions in the solution landscapes in Fig. \ref{b1.25}, we can draw the bifurcation diagram as a function of $\lambda^2$ in Fig. \ref{bifurcation1.25}. The solution branches has two families. One family is bifurcated from BD-S via a pitchfork bifurcation, while the other family is bifurcated from X and BD-L via a saddle-node bifurcation. The S, L, and X class in Fig. \ref{b1.25} are bifurcated from BD-S, BD-L and X solutions in the bifurcation diagram, respectively. The BD-S is the unique stable state for a small $\lambda$. At $\lambda^2 =12$, this BD-S loses its stability and bifurcates into a stable D state; at $\lambda^2=18$, the index-$1$ BD-S bifurcates into an index-$2$ BD-S and an index-$1$ R-S. As $\lambda^2$ increases further, the index-$1$ R-S gains stability and bifurcates into an index-$1$ J-S. At $\lambda^2 = 16$, the X and BD-L emerge from a saddle-node bifurcation, i.e., they emerge suddenly without connecting with other branches. At $\lambda^2 = 16$, the index-$2$ X bifurcates into an index-$3$ X and an index-$2$ H; at $\lambda^2 = 17$, the index-$3$ X bifurcates into an index-$4$ X and an index-$3$ C. When $\lambda^2$ increases further, the index-$3$ C bifurcates into an index-$2$ C and an index-$3$ T at $\lambda^2 =21.2$. For the BD-L branch, the index-$1$ BD-L bifurcates into an index-$2$ BD-L and an index-$1$ R-L at $\lambda^2 = 18$. The index-$1$ R-L gains its stability and bifurcates into an index-$1$ J-L when $\lambda^2 \geqslant 30.2$. More stable and unstable solutions can be obtained if we continue to increase $\lambda^2$.
%As $\lambda^2$ continues to increase, we can find more stable and unstable states. But untill $\lambda^2 = 50$ there is no solution both in BD-L and BD-S families. 

In \cite{fang2020surface}, the authors studied bifurcations of the typical solutions BD-L, BD-S, R-L, R-S and D, and raised the question: What is the relationship between R-L and BD-S? Now we can answer this question. R-L is bifurcated from the BD-L branch, which emerges from a saddle-node bifurcation for small $b$. R-L and BD-S are on two different solution families; thus, they cannot be connected in the bifurcation diagram.

\begin{figure}[hbtp]
    \centering
    \includegraphics[width=\textwidth]{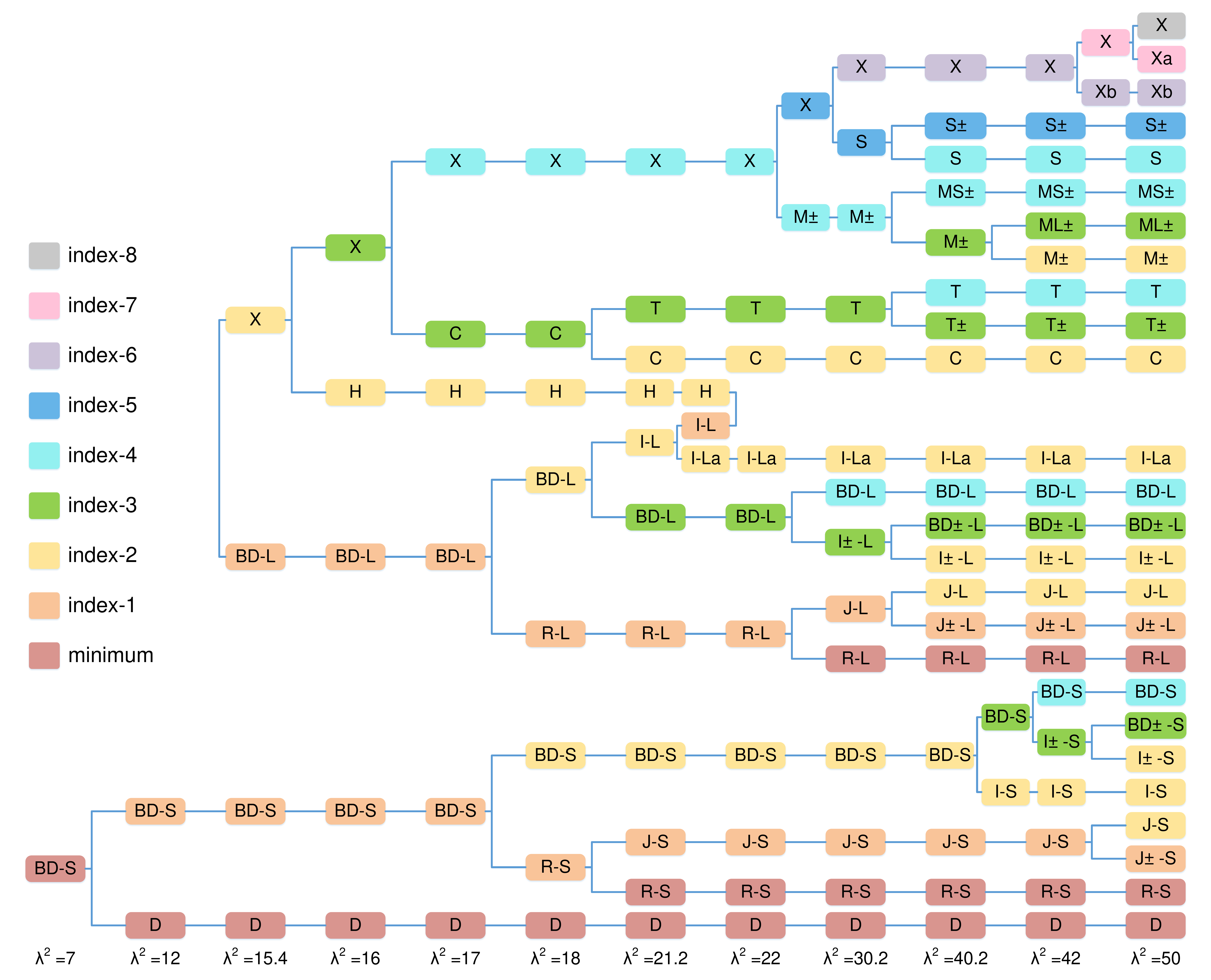}
    \caption{
      Bifurcation diagram as a function of $\lambda^2$ on a rectangle with $b=1.25$. Each small rectangle represents a solution as shown in Fig. \ref{b1.25}. The color of the small rectangle represents the index of the saddle point.}
    \label{bifurcation1.25}
\end{figure}

\subsection{Bifurcation diagram as a function of $b$}

We next study the bifurcation diagram as a function of the aspect ratio $b$. In Fig. \ref{Figure energy vs b}(a), we track the typical solutions on a rectangle in Fig. \ref{fig:typical} and C as $b$ changes with a fixed $\lambda^2=24$ . As $b$ increases, the energy of X, BD-L, R-L, and C increases greatly and these solutions merge and disappear eventually, while BD-S, R-S, and D remain at a low energy level. The energy difference between R-S and D decreases, which is consistent with the result in \cite{lewis2014colloidal}. With a large enough $b$, i.e., a very long rectangle domain, the nematic directors of R-S and D almost align along the long edges of rectangle with less effects on the configurations in corners. 

The effect of geometrical aspect ratio $b$ on the energy of these typical solutions also indicates the effect of $b$ on the bifurcation diagram in Fig. \ref{Figure energy vs b}(b).
When $b=1$ (square), WORS is the parent state, and we have found multiple pairs of rotationally equivalent solutions, such as BD1 and BD2, H1 and H2, R1 and R2 \cite{yin2020construction}. When $b=1.1$, WORS loses the cross structure and relaxes to X. The symmetry-breaking of the rectangle brings about the loss of degeneracy between some rotationally equivalent solutions. At $b=1.2$, the index-1 R-S gains its stability and bifurcates into an index-1 J-S through a pitchfork bifurcation. After that, there is no bifurcation for the S branch with an increasing $b$. On the other hand, as $b$ increases, the indices of X and BD-L increase and bifurcate into new solutions through pitchfork bifurcation. The number of solutions reaches its maximum when $b = 1.5$. When $b>1.5$, the solution number begins to decrease due to saddle-node bifurcations between C and R-L, X and BD-L, T and I-La and a pitchfork bifurcation between S$\pm$ and BD-L. When $b=2$, only the S class exists in the bifurcation diagram.

Based on these numerical findings, our hypothesis is that, at small $\lambda$, the solutions in the BD-L and X classes will disappear as long as $b$ is sufficiently large. Our numerical results may lead to new control strategies for confined NLC systems, since we can filter the energetically unfavorable defect patterns that are confined in a rectangle by adjusting its short edge length and its aspect ratio.

\begin{figure}[hbtp]
    \centering
    \includegraphics[width=\textwidth]{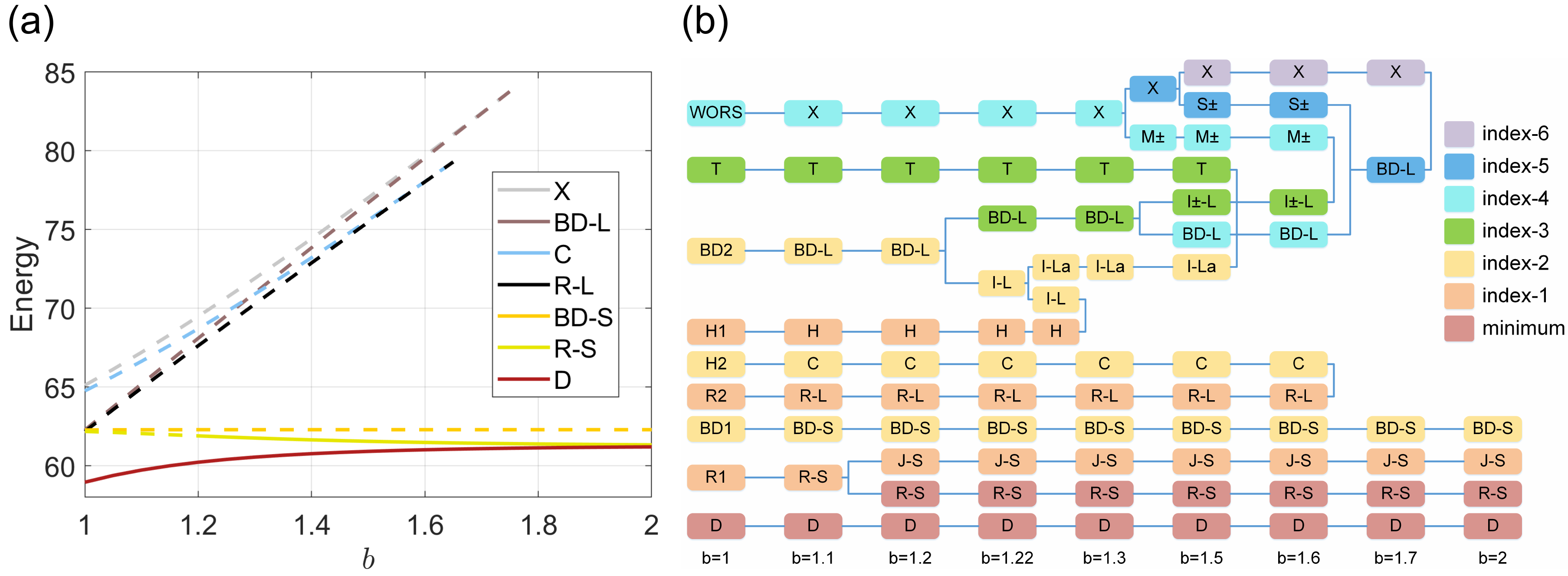}
    \caption{(a) The rLdG energy \eqref{energy_final} of some crucial states versus the aspect ratio $b$ with $\lambda^2=24$. The X state is the parent state of the solution landscape; BD-L(BD-S) is the state with highest index in the L(S) class; R-L, R-S, and D are stable states. (b) Bifurcation diagram of the rLdG model with $\lambda^2=24$. The suffixes ``1'' and ``2'' in the name of the configurations represent rotationally equivalent solutions at $b=1$ (square).}
	\label{Figure energy vs b}
\end{figure}

\subsection{The solution landscapes and bifurcations with $b=1.5$}\label{Sec 1.5}
To further investigate the effect of geometrical anisotropy, we compute solution landscapes in a rectangle with the large aspect ratio $b=1.5$ in Fig. \ref{b1.5} and the corresponding bifurcation diagram as a function of $\lambda^2$ in Fig. \ref{bifurcation1.5}. The effect on the defects along the short edges in S-class solutions is negligible when only varying the aspect ratio $b$. Thus, compared with the case at $b=1.25$, there are no striking differences in S class except for a change in bifurcation points. As for the L and X classes, with the increase of $b$, both the length of line defects and the number of point defects along the long edges increases, which indicates more unstable directions, i.e. higher Morse indices, such as BD-L and the new exotic saddle-point solutions as shown in Fig. \ref{b1.5}. Thus, X and BD-L emerge from saddle-node bifurcation with higher indices compared with the case when $b=1.25$. Furthermore, there are more bifurcations in the L class, and the number of solutions in that class increases from $5$ to $10$ at $\lambda^2=50$ as $b$ increases from $1.25$ to $1.5$.

It is notable that the bifurcation type for the emergence of R-L, C, I-La, and T changes as $b$ increases. At $b = 1.25$, R-L and C emerge from X and BD-L, respectively, through pitchfork bifurcations. T and I-La also emerge from C and I-L through pitchfork bifurcations, respectively (Fig. \ref{bifurcation1.25}). All four solutions are in either the X branch or the BD-L branch. At $b = 1.5$, both the pair index-1 R-L and index-2 C, and the pair index-2 I-La and index-3 T, emerge from saddle-node bifurcations; therefore, they are disconnected with X and BD-L branches (Fig. \ref{bifurcation1.5}).

\begin{figure}[hbtp]
    \centering
    \includegraphics[width=.965\textwidth]{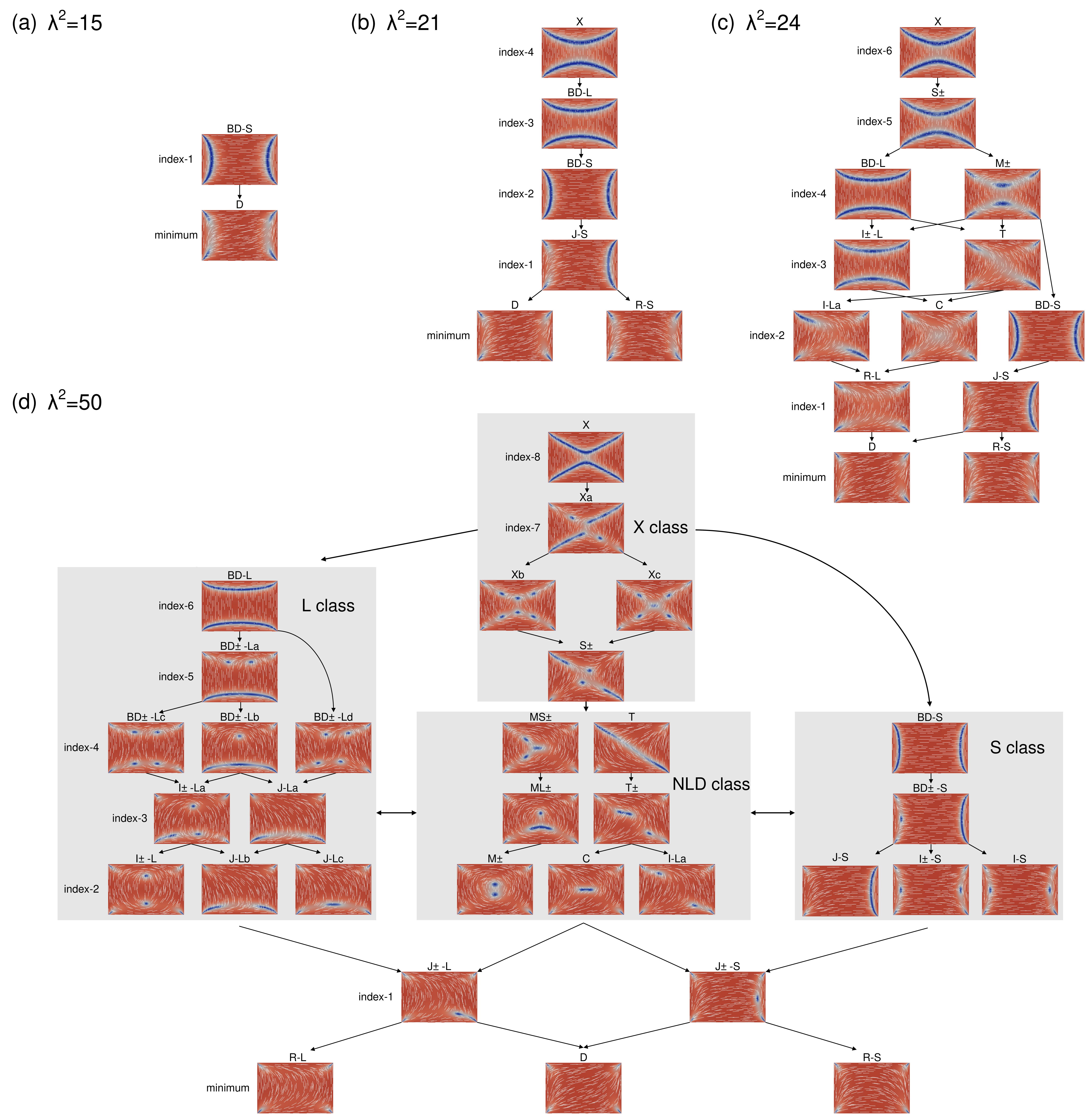}
	\caption{The solution landscapes at $b=1.5$ with (a) $\lambda^2=15$, (b) $\lambda^2=21$, (c) $\lambda^2=24$ and (d) $\lambda^2=50$.}
	\label{b1.5}
\end{figure}

%It is notable that the bifurcation type of the emergence of R-L, C, I-La and T changes. At $b = 1.25$, the R-L and C emerge from X and BD-L through the pitchfork bifurcation, respectively. Following that, T and I-La emerge from C and I-L through the pitchfork bifurcation, respectively (Fig. \ref{bifurcation1.25}). All of these four solutions are in the X branch or BD-L branch. At $b = 1.5$, the index-$1$ R-L and index-$2$ C, the index-$2$ I-La and index-$3$ T emerge from the saddle-node bifurcation, which are unconnected with X and BD-L branches. The solution landscape at $\lambda^2=24$ is shown in Fig. \ref{b1.5}(c).

\begin{figure}[hbtp]
    \centering
    \includegraphics[width=.95\textwidth]{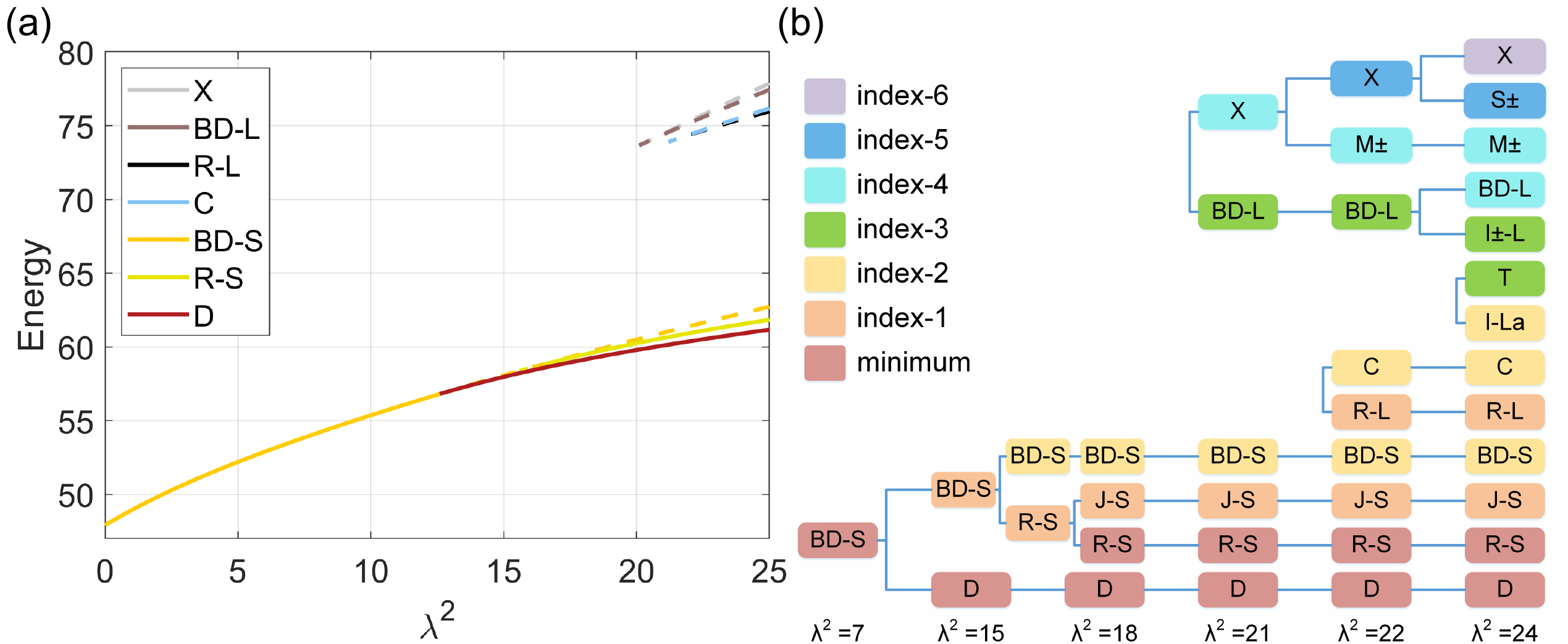}
    \caption{(a) The rLdG energy \eqref{energy_final} versus $\lambda^2$ with $b=1.5$. (b) Bifurcation diagram of a rLdG model on a rectangle with $b=1.25$. Each small rectangle represents a solution as shown in Fig. \ref{b1.5}. }
	\label{bifurcation1.5}
\end{figure}

%The solution landscape at $\lambda^2=50$ is shown in Fig. \ref{b1.5}(d). As $b$ increases, the length of the line defects and the number of point defects along the long edges increases, which indicates the more unstable directions, i.e. higher Morse indices, such as BD-L and new exotic saddle-point solutions. Thus, compared with the solution landscape at $b=1.25$, there are more bifurcations in L class and the number of solutions in L class increases from $5$ to $10$. Following the unstable eigendirection of the BD-L, one of the line defects splits into $\mp 1/2$ and $\pm 1/2$ interior point defects and converges to the index-$5$ BD$\pm$-La. Following the unstable eigendirection of the BD$\pm$-La, another line defect is symmetrically splits into $\mp 1/2$ and $\pm 1/2$ interior point defects (BD$\pm$-Lc) or splits into $\pm 1/2$ and $\mp 1/2$ interior point defects (BD$\pm$-Lc). BD$\pm$-Lb can be obtained by moving the $\mp 1/2$ point defect of the BD$\pm$-La to the near corner as a $\mp 1/4$ point defect. Searching for lower-index saddle points in L class from these three index-$4$ saddle points is the similar process as the interruption of line defects and the disappearance of point defects. As for the solutions in S class, the effect on the defects of them is negligible by varying $b$. Thus, compared with the solution landscape at $b=1.25$, there are almost no change in the S class.

\section{Conclusion and discussion}
We construct the solution landscapes and bifurcation diagrams of a rLdG model on a rectangle to study the effects of geometry on nematic equilibria at a fixed temperature $A=-\frac{B^2}{3C}$. We theoretically prove that when $\lambda$ is small enough for any $b$, or when $b$ is large enough for a fixed domain size, the rLdG system has a unique solution. This solution is BD-S, which features two line defects near the opposite short edges of a rectangular domain. It is worth noticing that the analogous solution of BD-S on a square, BD, is always unstable, which means that the geometric anisotropy can tune the stability of states. 

We systematically construct the solution landscapes of the rLdG free energy on rectangles with various values of $\lambda$ and $b$. Unlike the WORS, which is always the parent state on the square, the parent state on the rectangle, as domain size increases, changes from BD-S to X, whose line defects are near the rectangle's center region but lacks the cross structure of the WORS. Because of the symmetry breaking, energetically degenerate states on a square emerge as totally distinct states on a rectangle, such as BD-S and BD-L. Thus, compared with the square, there are more solutions in the solution landscapes; as well as connections between solutions are more complicated on a rectangle. We divide high-index solutions into four classes: X class, S class, L class, and NLD class according to the location of the defects and the connections of the solutions. Solutions in S(L) class feature point and/or line defects near the short(long) edges. Therefore, compared with the S class, the indices and the number of solutions in the L class change obviously as $b$ increases, since $b$ is related to the length of long edge. The solution landscape on the rectangle can be regarded as a mountain that has two main paths from the peak (X) to the foot (minima): X class$\rightarrow$L class$\rightarrow$J$\pm$-L$\rightarrow$minima and X class$\rightarrow$ S class$\rightarrow$J$\pm$-S$\rightarrow$minima. For a fixed rectangular domain size, with a higher aspect ratio, the transition pathway between two stable D states has a lower energy barrier. This indicates that larger geometrical anisotropy is more advantageous for switching between bistable D states, which could practically impact the design of bistable liquid crystal devices. This is an interesting example of the effect of geometrical anisotropy on confined defect patterns.

We present the bifurcation diagrams by tracking individual solution branches in the solution landscapes with various $\lambda$ and $b$ to investigate the emergence mechanisms of high-index solutions and the effects of geometrical anisotropy on bifurcation behaviors. For small $b$ ($b \leqslant 1.25$),  seen in the bifurcation diagram as a function of $\lambda^2$, we have three main branches: X branch, BD-L branch, and BD-S branch; i.e., most solutions are bifurcated from these three solutions through one or more pitchfork bifurcations. As $b$ increases, the bifurcation type for the emergence of some solutions change from pitchfork to saddle-node bifurcation, including the stable R-L state, which emerges from the latter bifurcation with a C state at large $b$. As $b$ increases more, the X and BD-L branches disappear via the saddle-node bifurcation and the number of solutions decreases rapidly. In particular, at $\lambda^2=24$, only the S class exists in the solution landscape when $b=2$. This finding suggests a new control strategy for confined NLC systems: we can adjust the short edge length and the aspect ratio to avoid energetically unfavorable defect patterns confined in a rectangle.
%However, not for all the value of $\lambda$, only S class left with the disappearance of X and BD-L. For example, at $\lambda^2=50$, the X and BD-L disappear at $b\approx 4.6$, but the R-L and C exist at least with $b\leqslant 5$. We can expect that more solutions will exist for large $b$ with higher value of $\lambda^2$. 

The results of this paper suggest several pertinent questions. Can we obtain the bounds for Morse indices of the critical points? If we can estimate such bounds, we may theoretically prove the disappearance of X and BD-L at large $b$. Moreover, although our results are conceived within the framework of a rLdG model with two degrees of freedom on 2D confinement, they are candidates for LdG energy critical points in a thin film. They also exist as a cross-section in 3D confinement, such as z-invariant solutions. What is more, there exist other physically relevant solutions in 3D NLC systems. For example, in a cylinder, we have a 3D critical point that vary across the cylinder height \cite{han2019transition}. Another example is a mixed 3D solution that interpolates between two distinct 2D critical points in the 3D cube \cite{canevari_majumdar_wang_harris}. Using these 2D critical points as a solution database, we can symmetrically investigate rich, exotic 3D solutions and the relation between them and 2D ones. We intend to pursue the solution landscape of confined 3D NLC systems in the future.

\section*{Acknowledgments}
We would like to thank Dr. Jianyuan Yin for helpful discussions and Dr. Lu Klinger for polishing the paper. Y. Han gratefully acknowledges the support from the Royal Society Newton International Fellowship.

\bibliographystyle{unsrt}
\bibliography{references}
\end{document}